\documentclass{llncs}
\usepackage{makeidx}  % allows for indexgeneration
\usepackage{graphicx}
\usepackage{algorithm}
\usepackage{algorithmic}
\usepackage{amsmath}

\newtheorem{step}{\indent Step}

\newtheorem{drul}{\indent Reduction Rule}
\newtheorem{cor}[theorem]{Corollary}
\begin{document}
\mainmatter          % for the preliminaries
\pagestyle{headings}  % switches on printing of running heads
\addtocmark{} % additional mark in the TO\underbrace{\underbrace{\underbrace{}}}C
\title{Kernelization and Parameterized Algorithms for 3-Path Vertex Cover \thanks{This is the version accepted by TAMC 2016.
To appear in: TAMC 2016, LNCS 9796, pp. 1--15, 2016.}
}
\titlerunning{$3$-Path Vertex Cover}  % abbreviated title (for running head)
% \author{Mingyu Xiao\inst{1} \and Shaowei Kou\inst{2}}
\author{Mingyu Xiao \and Shaowei Kou}
\authorrunning{ M. Xiao and S. Kou} % abbreviated author list (for running head)
\tocauthor{}
\institute{School of Computer Science and Engineering, \\
University of Electronic Science and Technology of China,
Chengdu 611731, China\\
\email{myxiao@gmail.com, kou\_sw@163.com}
%\and
%University of Electronic Science and Technology of China, School of Computer Science and Engineering, Cheng du, China
}
\maketitle              % typeset the title of the contribution
\begin{abstract}
A 3-path vertex cover in a graph is a vertex subset $C$ such that every path of three vertices contains at least one vertex from $C$. The parameterized 3-path vertex cover problem asks whether a graph has a 3-path vertex cover of size at most $k$. In this paper, we give a kernel of $5k$ vertices and an $O^*(1.7485^k)$-time polynomial-space
algorithm for this problem, both new results improve previous known bounds.

%\keywords{Graph algorithms, Kernelization, Parameterized algorithms, 3-path vertex cover}
\end{abstract}

\section{Introduction}
A vertex subset $C$ in a graph is called an \emph{$\ell$-path vertex cover} if every path of $\ell$ vertices in the graph contains at least one vertex from $C$. The $\ell$-path vertex cover problem, to find
an $\ell$-path vertex cover of minimum size, has been studied in the literature~\cite{Bresar:kpath1,Bresar:kpath2}.
When $\ell=2$, this problem becomes the famous vertex cover problem and it has been well studied.
In this paper we study the $3$-path vertex cover problem.
A 3-path vertex cover is also known as a \emph{$1$-degree-bounded deletion set}.
The $d$-degree-bounded deletion problem~\cite{Fellows:BDD,Xiao:BDD,X:dbound} is to delete a minimum number of vertices from a graph such that the
remaining graph has degree at most $d$.
The 3-path vertex cover problem is exactly the $1$-degree-bounded deletion problem.
Several applications of 3-path vertex covers have been proposed in~\cite{Bresar:kpath2,Kardos:3path,Xiao:faw2015}.

It is not hard to establish the NP-hardness of the 3-path vertex cover problem by reduction from the vertex cover problem. In fact, it remains NP-hard even in planar graphs~\cite{Yannakakis} and in $C_4$-free bipartite graphs with vertex degree at most 3~\cite{Boliac:dissociation number}.
There are several graph classes, in which the problem can be solved in polynomial time~\cite{Alekseev,Asdre:pathcover,Boliac:dissociation number,Bresar:kpath2,Cameron:Independentpackings,Goring,Hung:pathcover2,Lozin,Orlovich:dissociation set,Papadimitriou}.

The $3$-path vertex cover problem has been studied from approximation algorithms, exact algorithms and
parameterized algorithms.
 There is a randomized approximation algorithm with an expected approximation ratio of $\frac{23}{11}$~\cite{Kardos:3path}.
 In terms of exact algorithms,
Kardo\v{s} et al. \cite{Kardos:3path} gave an $O^*(1.5171^n)$-time algorithm to compute a maximum dissociation set
in an $n$-vertex graph. Chang et al.~\cite{Chang:BDD1set} gave an $O^*(1.4658^n)$-time algorithm and the result was further improved to $O^*(1.3659^n)$ later~\cite{Xiao:faw2015}.

In parameterized complexity, this problem is fixed-parameter tractable by taking the size $k$
of the 3-path vertex cover
as the parameter.
The running time bound of parameterized algorithm for this problem has been improved at least three times
during the last one year. Tu~\cite{tu:p3} showed that the problem can be solved in $O^*(2^k)$ time.
Wu~\cite{wu:p3} improved the result to $O^*(1.882^k)$ by using the measure-and-conquer method.
The current best result is $O^*(1.8172^k)$ by Katreni\v{c}~\cite{Katrenic:3PVC}.
In this paper we will further improve the bound to  $O^*(1.7485^k)$.

Another important issue in parameterized complexity is kernelization.
A kernelization algorithm is a polynomial-time algorithm which, for an input graph with a parameter $(G,k)$ either concludes that $G$ has no 3-path vertex cover of size $k$ or returns an equivalent instance $(G',k')$, called a \emph{kernel}, such that $k'\leq k$ and the
size of $G'$ is bounded by a function of $k$.
Kernelization for the $d$-degree-bounded deletion problem has been studied in the literature~\cite{Fellows:BDD,Xiao:BDD}.
For $d=1$, Fellows et al.'s algorithm~\cite{Fellows:BDD} implies a kernel of $15k$ vertices for the 3-path vertex cover problem,
and Xiao's algorithm~\cite{Xiao:BDD} implies a kernel of $13k$ vertices.
There is another closed related problem, called the \emph{3-path packing} problem.
In this problem, we are going to check if a graph has a set of at least $k$ vertex-disjoint 3-paths.
% In this problem, we are going to check if there exists a collection of at least $k$ vertex-disjoint 3-paths in a graph.
When we discuss kernelization algorithms, most structural properties of the 3-path vertex cover problem
and the 3-path packing problem are similar. Several previous kernelization algorithms for the
 3-path packing problem are possible to be modified for the 3-path vertex cover problem.
 The bound of the kernel size of the 3-path packing problem has been improved for several times from the first
 bound of $15k$~\cite{Prieto:stars} to $7k$~\cite{Wang:7k} and then to $6k$~\cite{Chen:packing}.
 Recently, there is a paper claiming a bound of $5k$ vertices for the
 3-path packing problem  in net-free graphs~\cite{Chang:5kkernel}. Although the paper~\cite{Chang:5kkernel} provides some useful ideas,
 the proof in it is incomplete and the algorithm may not stop.
% \footnote{The algorithm in~\cite{Chang:5kkernel} contains several rules. However,
%after applying a rule, previous rules may become applicable again.
%There are examples where the algorithm will not stop.}.
Several techniques for the
 3-path packing problem in~\cite{Wang:7k} and~\cite{Chang:5kkernel} will be used in our kernelization algorithm.
 We will give a kernel of $5k$ vertices for the 3-path vertex cover problem.

Omitted proofs in this extended abstract can be found in the full version of this paper.
% Due to limitation on the space,
 %the parameterized algorithm and the proof of an important lemma in the kernelization algorithm are moved to Appendix.

%In this paper, we show the following result.
%\begin{theorem}
%There is an algorithm that, given an instance $(G,k)$, either reports that $G$ has no 3-path vertex cover of size $k$ or produces an equivalent instance with at most $5k$ vertices.
%The algorithm runs in $O(n^c)$ time, where $n$ is the number of vertices of $G$ and $c$ is a constant value.
%\end{theorem}

%\textbf{Organization of paper.}
%In Section~\ref{sec pre} we give some preliminaries.
%In Section 3 we give the overall framework of our kernelizaiton algorithms.
%In Section 4 we give some preprocess rule for reduction, and prove the correctness of them.
%Section 5 describe the reduction method of our algorithm.
%Section 6 gives the size of the kernel.

\section{Preliminaries}\label{sec pre}
We let $G=(V,E)$ denote a simple and undirected graph with $n=|V|$  vertices and $m=|E|$ edges.
A singleton $\{v\}$ may be simply denoted by $v$.
The vertex set and edge set of a graph $G'$ are denoted by $V(G')$ and $E(G')$, respectively.  For a subgraph (resp., a vertex subset) $X$, the subgraph induced by $V(X)$ (resp., $X$) is simply denoted by $G[X]$, and $G[V\setminus V(X)]$ (resp., $G[V\setminus X]$) is also written as $G\setminus X$.
A vertex in a subgraph or a vertex subset $X$ is also called a \emph{$X$-vertex}.
For a vertex subset $X$, let $N(X)$ denote the set of \emph{open neighbors} of $X$, i.e., the vertices in $V\setminus X$ adjacent to some vertex in $X$, and  $N[X]$ denote the set of \emph{closed neighbors} of $X$, ie., $N(X)\cup X$. The \emph{degree} of a vertex $v$ in a graph $G$, denoted by $d(v)$, is defined to be the number of vertices adjacent to $v$ in $G$.
Two vertex-disjoint subgraphs $X_1$ and $X_2$ are \emph{adjacent} if there is an edge with one endpoint in $X_1$ and the other in $X_2$.
%A vertex subset $D \subseteq V$ of $G=(V,E)$ is called a \emph{bound-degree-1 set} if the maximum degree of $G[D]$ is at most one.
The number of connected components in a graph $G$ is denoted by $Comp(G)$ and the number of
connected components of size $i$ in a graph $G$ is denoted by $Comp_i(G)$. thus, $Comp(G)=\sum_i Comp_i(G)$.

A 3-\emph{path}, denoted by $P_3$, is a simple path with three vertices and two edges.
A vertex subset $C$ is called a \emph{3-path vertex cover} or a $P_3VC$-$set$ if there is no 3-path in $G \setminus C$.
Given a graph $G=(V,E)$, a $P_3$-$packing$ $\mathcal{P} = \{L_1,L_2,...,L_t\}$ of size $t$ is a collection of vertex-disjoint $P_3$ in $G$, i.e., each element $L_i \in \mathcal{P}$ is a 3-path in $G$ and $V(L_{i_1}) \cap V(L_{i_2}) = \emptyset$ for any two different 3-paths $L_{i_i}, L_{i_2} \in \mathcal{P}$.
A $P_3$-packing is \emph{maximal} if it is not properly contained in any strictly larger $P_3$-packing in $G$.
The set of vertices in 3-paths in $\mathcal{P}$ is denoted by $V(\mathcal{P})$.

Let $\mathcal{P}$ be a $P_3$-packing and $A$ be a vertex set such that
$A\cap V(\mathcal{P})=\emptyset$ and $A$ induces a graph of maximum degree 1.
We use $A_i$ to denote the set of degree-$i$ vertices in the induced graph $G[A]$
for $i=0,1$. A component of two vertices in $G[A]$ is called an $A_1$-edge.
For each $L_i \in \mathcal{P}$, we use $A(L_i)$ to denote the set of $A$-vertices that are in the components of $G[A]$ adjacent to $L_i$.
For a 3-path $L_i \in \mathcal{P}$, the degree-2 vertex in it is called the \emph{middle vertex} of it and the two
degree-1 vertices in it are call the \emph{ending vertices} of it.

\section{A Parameterized Algorithm}
In this section we will design a parameterized algorithm for the 3-path vertex cover problem.
Our algorithm is a branch-and-reduce algorithm that runs in $O^*(1.7485^k)$ time and polynomial space,
improving all previous results.
In branch-and-reduce algorithms, the exponential part of the running time is determined by the branching operations in the algorithm. In a branching operation, the algorithm solves the current instance $I$ by solving several smaller instances. We will use the parameter $k$ as the measure of the instance and use $T(k)$ to denote the maximum size of the search tree generated by the algorithm running on any instance with parameter at most $k$.
A branching operation, which
generates $l$ small branches with measure decrease in the $i$-th branch being at least $c_i$,
creates a recurrence relation
$T(k) \le T(k - c_1) + T(k - c_2) + \cdots + T(k - c_l)+1$.
The largest root of the function $f(x) = 1 - \sum_{i=1}^{l} x^{-c_i}$ is called the \emph{branching factor} of the recurrence. Let $\gamma$ be the maximum branching factor among all branching factors in the algorithm.
The running time of the algorithm is bounded by $O^*(\gamma^k)$.
More details about the analysis and how to solve recurrences can be found in the monograph~\cite{Fomin:book}.
Next, we first introduce our branching rules and then present our algorithm.

\subsection{Branching Rules}
We have four branching rules. The first branching rule is simple and easy to observe.

\noindent\textbf{Branching rule (B1):} \emph{Branch on a vertex $v$ to generate $|N[v]| + 1$ branches by either\\
(i) deleting $v$ from the graph, including it to the solution set, and decreasing $k$ by 1, or\\
(ii) deleting $N[v]$ from the graph, including $N(v)$ to the solution set, and decreasing $k$ by $|N(v)|$,    or\\
(iii) for each neighbor $u$ of $v$, deleting $N[\{u,v\}]$ from the graph, including $N(\{u,v\})$ to the
solution set, and decreasing $k$ by $|N(\{u,v\})|$.}

A vertex $v$ is \emph{dominated} by a neighbor $u$ of it if $v$ is adjacent to all neighbors of $u$.
The following property of dominated vertices has been proved and used in~\cite{Xiao:faw2015}.%,Gupta:rregular}.

\begin{lemma} \label{lem_dominate}
Let $v$ be a vertex dominated by $u$. If there is a minimum 3-path vertex cover $C$ not containing $v$, then there is a minimum 3-path vertex cover $C'$ of $G$ such that $v,u \notin C'$ and $N(\{u,v\}) \subseteq C'$.
\end{lemma}

Based on this lemma, we design the following branching rule.

\medskip
\noindent\textbf{Branching rule (B2):} \emph{Branch on a vertex $v$ dominated by another vertex $u$ to generate
two instances by either\\
(i) deleting $v$ from the graph, including it to the solution set, and decreasing $k$ by 1, or\\
(ii) deleting $N[\{u,v\}]$ from the graph, including $N(\{u,v\})$ to the solution set, and decreasing $k$ by
$|N(\{u,v\})|=|N(v)|-1$.
}

For a vertex $v$, a vertex $s\in N_2(v)$ is called a \emph{satellite} of $v$ if there is a neighbor $p$ of $v$ such that $N[p]-N[v]=\{s\}$. The vertex $p$ is also called the \emph{parent} of the satellite $s$ at $v$.

\begin{lemma} \label{lem_sat}
Let $v$ be a vertex that is not dominated by any other vertex.
If $v$ has a satellite,
then there is a minimum 3-path vertex cover $C$ such that either $v \in C$ or $v,u \not\in C$ for a neighbor $u$ of $v$.
\end{lemma}
%\begin{proof}
%Assume that there is a minimum 3-path vertex cover $C$ such that $v \notin C$ and $N(v) \subseteq C$.
%We can see that $C' = C \cup \{s\} \setminus \{u\}$ is a minimum 3-path vertex cover not containing $v$ and a neighbor %$u$ of it.
%\hfill \qed
%\end{proof}

%Based on lemma~\ref{lem_sat}, we design the following branching rule.

\medskip
\noindent\textbf{Branching rule (B3):} \emph{Let $v$ be a vertex that has a satellite but is not dominated by any other vertex.
Branch on $v$ to generate $|N[v]|$ instances by either\\
(i) deleting $v$ from the graph, including it to the solution set, and decreasing $k$ by 1, or\\
(ii) for each neighbor $u$ of $v$, deleting $N[\{u,v\}]$ from the graph, including $N(\{u,v\})$ to the solution set,
and decreasing $k$ by $|N(\{u,v\})|$.
}

\begin{lemma} \label{lem_triangle}
Let $v$ be a degree-3 vertex with a degree-1 neighbor $u_1$ and two adjacent neighbors $u_2$ and $u_3$.
There is a minimum 3-path vertex cover $C$ such that either $C\cup \{u_1, v\}=\emptyset$ or $C\cup \{u_1, u_2,u_3\}=\emptyset$.
\end{lemma}
%\begin{proof}
%Let $C$ be  a minimum 3-path vertex cover.
%If $C\cup \{u_2,u_3\}=\emptyset$, then $v$ must be in $C$ to cover a 3-path $vu_2u_3$. For this case, $u_1$ is not in %$C$ otherwise $C\setminus\{u_1\}$ would be a 3-path vertex cover of smaller size. We know that $C\cup \{u_1, %u_2,u_3\}=\emptyset$.
%If $C$ contains only one of $u_2$ and $u_3$, say $u_3$, then $C$ must contain another vertex in the 3-path in %$u_1vu_3$. So $C$ contains two vertices in $\{v,u_1,u_2,u_3\}$. We can see that $C'=C\setminus \{v,u_1,u_2,u_3\} \cup
%\{u_2,u_3\}$ is still a minimum 3-path vertex cover, which does not contain $u_1$ and $v$.
%\hfill \qed
%\end{proof}

%Based on lemma~\ref{lem_triangle}, we design the following branching rule.

\medskip
\noindent\textbf{Branching rule (B4):} \emph{ Let $v$ be a degree-3 vertex with a degree-1 neighbor $u_1$ and two adjacent neighbors $u_2$ and $u_3$.
Branch on $v$ to generate two instances by either\\
(i) deleting $N[\{u_1,v\}]$ from the graph, including $\{u_2,u_3\}$ to the solution set, and decreasing $k$ by 2, or\\
(ii) deleting $N[\{u_2,u_3\}]\cup \{u_1\}$ from the graph, including $N(\{u_2,u_3\})$ to the solution set, and decreasing $k$ by $|N(\{u_2,u_3\})|$.
}

\subsection{The Algorithm}
We will use ${\tt P3VC}(G,k)$ to denote our parameterized algorithm.
The algorithm contains 7 steps. When we execute one step, we assume that all previous steps are not applicable anymore on the current graph. We will analyze each step after describing it.

\begin{step}[Trivial cases]\label{step-trivial}
If $k\leq 0$ or the graph is an empty graph, then return the result directly. If the graph has a component of maximum degree 2, find a minimum 3-path vertex cover $S$ of it directly, delete this component from the graph, and decrease $k$ by the size of $S$.
\end{step}

%\invis{
%\begin{step}[Trivial cases]
%If $k< 0$, then return `\textbf{no}'; else if $G$ is an empty graph, then return `\textbf{yes}'.
%\end{step}
%
%\begin{step}[Independent vertices]
%If there is a vertex $v$ with $d(v) = 0$, then return
%${\tt aim}(G\setminus \{v\},k-1)$.
%\end{step}
%
%\begin{step}[Independent edges]
%If there are two adjacent degree-1 vertices $v$ and $u$, then return
%${\tt aim}(G\setminus \{v,u\},k)$.
%\end{step}
%
%\begin{step}[Components of cycles]
%If there is a connected component of the graph such that each vertex in it is a degree-2 vertex, then select an arbitrary vertex $v$ in this component and return
%${\tt aim}(G\setminus \{v\},k-1)$.
%\end{step}
%}

After Step~\ref{step-trivial}, each component of the graph contains at least four vertices.
A degree-1 vertex $v$ is called a \emph{tail} if its neighbor $u$ is a degree-2 vertex.
Let $v$ be a tail, $u$ be the degree-2 neighbor of $v$, and $w$ be the other neighbor of $u$.
We show that there is a minimum 3-path vertex cover containing $w$ but not containing any of $u$ and $v$.
At most one of $u$ and $v$ is contained in any minimum 3-path vertex cover $C$, otherwise $C \cup \{w\} \setminus \{u,v\}$ would be a smaller 3-path vertex cover.
If none of $u$ and $v$ is in a  minimum 3-path vertex cover $C$, then $w$ must be in $C$ to cover the 3-path $uvw$ and then $C$ is a claimed  minimum 3-path vertex cover.
If exactly one of $u$ and $v$ is contained in a minimum 3-path vertex cover $C$, then $C'=C \cup \{w\} \setminus \{u,v\}$ is a claimed minimum 3-path vertex cover.

\begin{step}[Tails] \label{step_tail}
If there is a degree-1 vertex $v$ with a degree-2 neighbor $u$, then return
${\tt p3vc}(G\setminus N[\{v,u\}],k-1)$.
\end{step}

%The first 4 steps can not deal with all degree-1 vertices. The remaining graph may still have a degree-1 vertex with a neighbor of degree $\geq 3$.

\begin{step}[Dominated vertices of degree $\ge 3$] \label{step_dominate}
If there is a vertex $v$ of degree $\geq 3$ dominated by $u$, then branch on $v$ with Rule~(B2) to generate two branches
$${\tt p3vc}(G \setminus \{v\},k-1) \quad \mbox{and} \quad {\tt p3vc}(G \setminus N[\{v,u\}],k-|N(\{v,u\})|).$$
\end{step}
Lemma~\ref{lem_dominate} guarantees the correctness of this step.
Note that $|N(\{v,u\})|= d(v) - 1$. This step gives a recurrence
\begin{equation}
T(k) \le T(k - 1) + T(k - (d(v)-1)) + 1,
\end{equation}
where $d(v) \ge 3$. For the worst case that $d(v)=3$, the branching factor of it is 1.6181.

A degree-1 vertex with a degree-1 neighbor will be handled in Step~\ref{step-trivial},
a degree-1 vertex with a degree-2 neighbor will be handled in Step~\ref{step_tail}, and
a degree-1 vertex with a neighbor of degree $\geq 3$ will be handled in Step~\ref{step_dominate}.
So after Step~\ref{step_dominate}, the graph has no vertex of degree $\leq 1$.
Next we consider degree$\ge 4$ vertices.

\begin{step}[Vertices of degree $\ge 4$ with satellites] \label{step_sattellite}
If there is a vertex $v$ of $d(v) \ge 4$ having a satellite, then branch on $v$ with Rule (B3) to generate $d(v)+1$ branches
$$ {\tt p3vc}(G \setminus \{v\},k-1)\quad \mbox{and} \quad
{\tt p3vc}(G\setminus N[\{v,u\}],k-|N(\{v,u\})|)\mbox{~for each $u\in N(v)$}.$$
\end{step}

The correctness of this step is guaranteed by lemma~\ref{lem_sat}.
Note that there is no dominated vertex after Step~\ref{step_dominate}. Each neighbor $u$ of $v$ is
adjacent to at least one vertex in $N_2(v)$ and then $|N(\{v,u\})|\geq d(v)$.

This step gives a recurrence
\begin{equation}
T(k) \le T(k - 1) + d(v) \cdot T(k - d(v)) + 1,
\end{equation}
where $d(v)\geq 4$.
For the worst case that $d(v) = 4$, the branching factor of it is 1.7485.

After Step~\ref{step_sattellite}, if there is still a vertex of degree $\geq 4$, we use
the following branching rule. Note that now each neighbor $u$ of $v$ is
adjacent to at least two vertices in $N_2(v)$ and then $|N(\{v,u\})|\geq d(v)+1$.

\begin{step}[Normal vertices of degree $\ge 4$ ] \label{step_normal}
If there is a vertex $v$ of $d(v) \ge 4$, then branch on $v$ with Rule (B1) to generate $d(v) + 2$ branches
\[
\begin{split}
{\tt p3vc}(G \setminus \{v\},k-1), &\quad {\tt p3vc}(G\setminus N[v],k-|N(v)|)\\
  \mbox{and}  & \quad {\tt p3vc}(G\setminus N[\{v,u\}],k-|N(\{v,u\})|)\mbox{~for each $u\in N(v)$}.
\end{split}
\]

\end{step}

Since $|N(\{v,u\})|\geq d(v)+1$, this step gives a recurrence
\begin{equation}
T(k) \le T(k - 1) + T(k - d(v)) + d(v) \cdot T(k - (d(v)+1)) + 1,
\end{equation}
which $d(v)\geq 4$.
For the worst case that $d(v) = 4$, the branching factor of it is 1.6930.

\medskip
After Step~\ref{step_normal}, the graph has only degree-2 and degree-3 vertices.
We first consider degree-2 vertices.

A path $u_0u_1u_2u_3$ of four vertices is called a \emph{chain} if the first vertex $u_0$ is of degree $\geq 3$ and the two middle vertices are of degree 2.
Note that there is no chain with $u_0=u_3$ after Step~\ref{step_dominate}. So when we discuss a chain we always assume
that $u_0\neq u_3$.
A chain can be found in linear time if it exists. In a chain $u_0u_1u_2u_3$, $u_2$ is a satellite of $u_0$ with a parent $u_1$.

%A path $u_0u_1u_2u_3$ of four vertices is called a \emph{short chain} if the first vertex $u_0$ and last vertex $u_3$ are of degree $\geq 3$ and the two middle vertices are of degree 2, where we allow $u_3=u_0$.

\begin{step}[Chains] \label{step_chain}
If there is a chain $u_0u_1u_2u_3$, then branch on $u_0$ with Rule~(B3). In the branch where $u_0$ is deleted and
included to the solution set, $u_1$ becomes a tail and we further handle the tail as we do in Step~\ref{step_tail}.

We get the following branches
\[
\begin{split}
  &\quad {\tt p3vc}(G \setminus N[\{u_1,u_2\}],k-2)\\
  \mbox{and}  &\quad {\tt p3vc}(G\setminus N[\{u_0,u\}],k-|N(\{u_0,u\})|) \mbox{~for each $u\in N(u_0)$}.
\end{split}
\]
\end{step}

Note that $|N(\{u_0,u\})|\geq d(u_0)$ since there is no dominated vertex. We get a recurrence
$$
T(k) \le T(k - 2) + d(u_0) \cdot T(k - d(u_0)) + 1,
$$
where $d(u_0) \ge 3$. For the worst case that $d(u_0) = 3$, the branching factor of it is 1.6717.

\medskip
After Step~\ref{step_chain}, each degree-2 vertex must have two nonadjacent degree-3 vertices.
Note that no degree-2 is in a triangle if there is no dominated vertex.

\begin{step}[Degree-2 vertices with a neighbor in a triangle] \label{step_special1}
If there is a degree-2 vertex $v$ with $N(v)=\{u,w\}$ such that a neighbor $u$ of it is in a triangle $uu_1u_2$,
 then branch on $w$ with Rule (B1) and then in the branch $w$ is deleted and included in the solution set further branch on $u$ with Rule (B4). We get the following branches
 \[
\begin{split}
{\tt p3vc}(G\setminus N[\{u,v\}],k-|N(\{u,v\})|),&\\
{\tt p3vc}(G \setminus N[\{u_1, u_2\}]\cup \{u,w\}, k - |N(\{u_1, u_2\})\cup \{w\}|),&~ \mbox{and} \\
{\tt p3vc}(G \setminus N[\{w,u'\}],k-|N(\{w,u'\})|) \mbox{~for each $u'\in N(w)$}.&
\end{split}
\]
\end{step}

There two neighbors $u$ and $w$ of $v$ are degree-3 vertices.
 Since there is no dominated vertex, for any edge $v_1v_2$ it holds $|N(\{v_1,v_2\})|\geq \min \{d(v_1), d(v_2)\}$.
 We know that $|N(\{u,v\})|\geq d(u)=3$, $|N(\{u_1, u_2\})\cup \{w\}|\geq |N(\{u_1, u_2\})|\geq 3$ (since no degree-2 vertex is in a triangle) and $|N(\{w,u'\})|\geq d(w)$ for each $u'\in N(w)$.
We get the following recurrence
$$
T(k) \le T(k - 3) + T(k - 3) + 3 \cdot T(k - 3) + 1.
$$
The branching factor of it is 1.7100.

\medskip
After Step~\ref{step_special1}, no degree-3 vertex in a triangle is adjacent to a degree-2 vertex.
%if a degree-2 $v$ is adjacent to a degree-3 vertex $u$, then there would not be a triangle adjacent to $u$.
%If $u$ has at least one neighbor of degree-3, the following step would be applicable.

\begin{step}[Degree-2 vertices $v$ with a degree-3 vertex in $N_2(v)$] \label{step_special2}
If there is a degree-2 vertex $v$ such that at least one of its neighbors $u$ and $w$, say $u$, has a degree-3 neighbor $u_1$, then branch on $u$ with  Rule (B1) and in the branch where $u$ is deleted and included to the solution set, branch on $w$ with Rule (B2).
We get the branches
\[
\begin{split}
{\tt p3vc}(G \setminus \{u, v, w\}, k - 2),~ {\tt p3vc}(G\setminus N[\{w,v\}],k-|N(\{w,v\})|),\\
~ \mbox{and} ~ {\tt p3vc}(G \setminus N[\{u,u'\}],k-|N(\{u,u'\})|) \mbox{~for each $u'\in N(u)$.}
\end{split}
\]
\end{step}
Note that $d(u)=d(w)=3$.
It holds $|N(\{w,v\})|\geq d(w)=3$ and $|N(\{u,u'\})|\geq d(u)=3$ for $u'\in N(u)$.
Furthermore, we have that $|N(\{u,u_1\})|\geq 4$ because $u$ and $u_1$ are degree-3 vertices not in any triangle.
We get the following recurrence
$$
T(k) \le T(k - 2) + T(k - 3) + 2 \cdot T(k - 3) + T(k - 4).
$$
The branching factor of it is 1.7456.

\begin{lemma} \label{correct}
After Step~\ref{step_special2}, if the graph is not an empty graph, then each component of the graph is
either a 3-regular graph or a bipartite graph with one side of degree-2 vertices and one side of degree-3 vertices.
\end{lemma}
%\begin{proof}
%Assume the graph is not an empty graph. Let $H$ be an arbitrary component of the graph.
%If $H$ has no degree-2 vertex, then it is a 3-regular graph since the graph has only degree-2 and degree-3 vertices.
%Note that if $H$ has no degree-3 vertex, then it should have be eliminated in the first step. Next, we assume that $H$ %contain both of degree-2 and degree-3 vertices.
%The graph $H$ has no two adjacent degree-2 vertices otherwise there would be a chain and Step~\ref{step_chain} would %be applied.
%If $H$ has two adjacent degree-3 vertices, then we can find a connected component $H'$ of degree-3 vertices in $H$ %such that $H'$ contains more than two vertices. For this case, $H'$ must be adjacent to a degree-2 vertex $v$ %otherwise
%$H'=H$ is a 3-regular graph, and then $v$ is a degree-2 vertex satisfying the conditions in
%either Step~\ref{step_special1} or Step~\ref{step_special2}. So $H$ has no two adjacent degree-3 vertices.
%Thus, $H$ is a bipartite graph with one side of degree-2 vertices and one side of degree-3 vertices.
%\hfill \qed
%\end{proof}

\begin{lemma} \label{bipartite}
Let $G=(V_1\cup V_2,E)$ be a bipartite graph such that all vertices in $V_1$ are of degree 2 and all vertices in $V_2$ are of degree 3.
The set $V_1$ is a minimum 3-path vertex cover of $G$.
\end{lemma}
%\begin{proof}
%Let $S$ be a minimum 3-path vertex cover of $G$.
%Let $S_1=S\cap V_1$, $S_2=S\cap V_2$, $U_1=V_1\setminus S_1$ and $U_2=V_2\setminus S_2$.
%Since $S$ is a 3-path vertex cover, we know that each vertex in $U_1$ is adjacent to at most one vertex in $U_2$.
%Thus, after deleting $S_1\cup U_2$ from the graph, each vertex in $U_1$ is of degree $\geq 2$ and each vertex in $S_2$ %is of degree $\leq 2$,
%which implies that $|U_1|\leq |S_2|$. We know that $|V_1|=|S_1|+|U_1|\leq |S_1|+|S_2|=|S|$.
%It is clear that $V_1=S_1\cup U_1$ is a 3-path vertex cover of $G$. Thus, $V_1$ is also a minimum 3-path vertex cover %of $G$.
%\hfill \qed
%\end{proof}

%The correctness of the following step is based on Lemma~\ref{bipartite}.
\begin{step}[Bipartite graphs] \label{step_bipartite}
If the graph has a component $H$ being a bipartite graph with one side $V_1$ of degree-2 vertices and one side $V_2$ of degree-3 vertices,
then return
${\tt p3vc}(G\setminus H,k-|V_1|)$.
\end{step}

\begin{step}[3-regular graphs] \label{step_regular}
If the graph is a 3-regular graph, pick up an arbitrary vertex $v$ and branch on it with Rule (B1).
\end{step}

Lemma~\ref{correct} shows that the above steps cover all the cases, which implies the correctness of the algorithm.
Note that all the branching operations except Step~\ref{step_regular} in the algorithm have a branching factor at most 1.7485.
We do not analyze the branching factor for Step~\ref{step_regular}, because this step will not exponentially increase the running time bound of our algorithm.
Any proper subgraph of a connected 3-regular graph is not a 3-regular graph.
For each connected component of a 3-regular graph, Step~\ref{step_regular} can be applied for at most one time and all other branching operations have a branching factor at most 1.7485.
Thus each connected component of a 3-regular graph can be solved in $O^*(1.7485^k)$ time. Before getting a connected component of a 3-regular graph, the algorithm
always branches with branching factors of at most 1.7485. Therefore,
\begin{theorem}
The 3-path vertex cover problem can be solved in $O^*(1.7485^k)$ time and polynomial space.
\end{theorem}

\section{Kernelization}
In this section, we show that the parameterized 3-path vertex cover problem allows a kernel with at most $5k$ vertices.

\subsection{Graph decompositions}
The kernelization algorithm is based on a vertex decomposition of the graph, called \emph{good decomposition},
which can be regarded as an extension of the crown decomposition~\cite{AbuKhzam:kernelization}.
 Based on a good decomposition we show that an optimal solution  to a special local part of the graph is contained in an optimal solution to the whole graph. Thus, once we find a good decomposition, we may be able to reduce the graph by
adding some vertices to the solution set directly. We only need to find good decompositions in polynomial time in
graphs with a large size to get problem kernels.
Some previous rules to kernels for the parameterized 3-path packing problem~\cite{Chen:packing,Fermau:packinglength2,Wang:7k} are adopted here to find good decompositions in an effective way.

\begin{definition}
A \emph{good decomposition} of a graph $G=(V,E)$ is a decomposition $(I, C, R)$ of the vertex set $V$ such that
\begin{enumerate}
\item the induced subgraph $G[I]$ has maximum degree at most 1;
\item the induced subgraph $G[I\cup C]$ has a $P_3$-packing of size $|C|$;
\item no vertex in $I$ is adjacent to a vertex in $R$.
\end{enumerate}
\end{definition}

\begin{lemma} \label{lemma_lemma1}
A graph $G$ that admits a good decomposition $(I,C,R)$ has a $P_3$-vertex cover (resp., $P_3$-packing) of size $k$ if and only if $G[R]$ has a $P_3$-vertex cover (resp., $P_3$-packing) of size $k-|C|$.
\end{lemma}

Lemma~\ref{lemma_lemma1} provides a way to reduce instances of the parameterized 3-path vertex cover problem
based on a good decomposition  $(I,C,R)$ of the graph: deleting $I\cup C$ from the graph and adding $C$ to the solution set. Here arise a question: how to effectively find good decompositions? It is strongly related to the
quality of our kernelization algorithm. The kernel size will be smaller if we can polynomially compute a good decomposition in a smaller graph.
%Some previous kernel algorithms for the parameterized 3-path packing problem use the following two lemmas.
Recall that we use $Comp(G')$ and $Comp_i(G')$ to denote the number of components and number of components with $i$ vertices in a graph $G'$, respectively.
For a vertex subset $A$ that induces a graph of maximum degree at most $1$ and $j=\{1,2\}$,
we use  $N_j(A)\subseteq N(A)$ to denote the set of vertices in $N(A)$ adjacent to at least one
component of size $j$ in $G[A]$,
and $N'_2(A)\subseteq N_2(A)$ be the set of vertices in $N(A)$ adjacent to at least one
component of size $2$ but no component of size 1 in $G[A]$.
We will use the following lemma to find good decompositions, which was also used in~\cite{Chang:5kkernel} to design kernel algorithms
for the 3-path packing problem.

\begin{lemma} \label{lemma_lemma2}
Let $A$ be a vertex subset of a graph $G$ such that each connected component of the induced graph $G[A]$ has at most 2 vertices.
%Let $N_1(A)\subseteq N(A)$ be the set of vertices in $N(A)$ adjacent to at least one single vertex in $G[A]$, and
%$N_2(A)\subseteq N(A)$ be the set of vertices in $N(A)$ adjacent to at least one vertex in a component of size 2 in $G[A]$,
%where it is possible that $N_1(A)\cap N_2(A)\neq \emptyset$.
If
\begin{eqnarray}Comp(G[A])> 2|N(A)|-|N'_2(A)|,\label{imp1}
\end{eqnarray}
then there is a good decomposition  $(I,C,R)$ of $G$ such that $\emptyset \neq I\subseteq A$ and $C\subseteq N(A)$.
Furthermore, the good decomposition $(I,C,R)$ together with
a $P_3$-packing of size $|C|$ in $G[I\cup C]$
can be computed in $O(\sqrt{n}m)$ time.
\end{lemma}

By using Lemma~\ref{lemma_lemma2}, we can get a linear kernel for the parameterized 3-path vertex cover problem
quickly. We find an arbitrary maximal $P_3$-packing $S$ and let $A=V\setminus V(S)$. We assume that $S$ contains
less than $k$ 3-paths and then $|V(S)|< 3k$, otherwise the problem is solved directly.
Note that $|N(A)|\subseteq |V(S)|$. If $|A| >12 k$, then $Comp(G[A])\geq {\frac{|A|}{2}} >6 k>  2|V(S)|\geq 2|N(A)|$ and we reduce the instance by Lemma~\ref{lemma_lemma2}. So we can get a kernel of $15k$ vertices.
This bound can be improved by using a special case of Lemma~\ref{lemma_lemma2}.

For a vertex subset $A$ such that $G[A]$ has maximum degree at most 1.
Let $A_0$ be the set of degree-1 vertices in $G[A]$.
Note that  $Comp(G[A_0])=Comp_2(G[A])$ and $|N(A_0)|=|N_2(A_0)|=|N_2(A)|$.
By applying  Lemma~\ref{lemma_lemma2}
on $A_0$, we can get
%a good decomposition if $Comp(G[A_0])> 2|N(A_0)|-|N_2(A_0)|$.
\setcounter{theorem}{0}
\begin{cor} \label{lemma_corollary}
Let $A$ be a vertex subset of a graph $G$ such that each connected component of the induced graph $G[A]$ has at most 2 vertices.
Let $N_2(A)\subseteq N(A)$ be the set of vertices in $N(A)$ adjacent to at least one vertex in a component of size 2 in $G[A]$.
If
\begin{eqnarray}Comp_2(G[A])> |N_2(A)|,\label{imp2}
\end{eqnarray}
then there is a good decomposition  $(I,C,R)$ of $G$ such that $\emptyset \neq I\subseteq A$ and $C\subseteq N(A)$.
Furthermore, the good decomposition $(I,C,R)$ together with
a $P_3$-packing of size $|C|$ in $G[I\cup C]$
can be computed in $O(\sqrt{n}m)$ time.
\end{cor}

Note that $|A|=Comp_1(G[A])+2\cdot Comp_2(G[A])$. If $|A|> 9k$, then $Comp_1(G[A])+2\cdot Comp_2(G[A])=|A|>9k
>  3|V(S)|\geq 3|N(A)|\geq (2|N(A)|-|N'_2(A)|)+|N_2(A)|$ and at least one of (\ref{imp1}) and (\ref{imp2}) holds.
Then by using Lemma~\ref{lemma_lemma2} and Corollary~\ref{lemma_corollary}, we can get a kernel of size $9k+3k=12k$.
It is possible to bound $|N(A)|$ by $k$ and then to get a kernel of size $3k+3k=6k$.
To further improve the kernel size to $5k$, we need some sophisticated techniques and deep analyses on the graph structure.

\subsection{A $5k$ kernel}
In this section, we use ``crucial partitions'' to find good partitions.
A vertex partition $(A,B,Z)$ of a graph is called a \emph{crucial partition} if it satisfies \emph{Basic Conditions} and
\emph{Extended Conditions}.
Basic Conditions include the following four items:
\begin{enumerate}
\item[(B1)] $A$ induces a graph of degree at most 1;
\item[(B2)] $B$ is the vertex set of a $P_3$-packing $\mathcal{P}$;
\item[(B3)] No vertex in $A$ is adjacent to a vertex in $Z$;
\item[(B4)] $|Z|\leq 5\cdot \gamma(G[Z])$, where $\gamma(G[Z])$ is the size of a minimum $P_3VC$-set in the induced subgraph $G[Z]$.
\end{enumerate}
Before presenting the definition of Extended Conditions, we give some used definitions.
We use $\mathcal{P}_j$ to denote the collection of 3-paths in $\mathcal{P}$ having $j$ vertices adjacent to $A$-vertices
$(j = 0,1,2,3)$. Then  $\mathcal{P} = \mathcal{P}_0 \cup \mathcal{P}_1 \cup \mathcal{P}_2 \cup \mathcal{P}_3$.
We use $\mathcal{P}^1$ to denote the collection of 3-paths $L\in \mathcal{P}$ such that $|A(L)|=1$.
%Each 3-path in $\mathcal{P}_1$ has exactly two free vertices.
We also partition $\mathcal{P}_1\setminus \mathcal{P}^1$ into two parts: \\
let $\mathcal{P}_M \subseteq \mathcal{P}_1\setminus \mathcal{P}^1$ be the collection of 3-paths with the middle vertex adjacent to some $A$-vertices; \\
%let $\mathcal{P}_L^1 \subseteq \mathcal{P}_1$ be the collection of 3-paths $L_i$ such that $|A(L_i)|=1$ and  one ending vertex of $L_i$ is adjacent to an $A$-vertex;\\
let $\mathcal{P}_L \subseteq \mathcal{P}_1\setminus \mathcal{P}^1$ be the collection of 3-paths $L_i$ such that $|A(L_i)|\geq 2$ and one ending vertex of $L_i$ is adjacent to some $A$-vertices.

A vertex in a 3-path in $\mathcal{P}$ is \emph{free} if it is not adjacent to any $A$-vertex.
A 3-path in $\mathcal{P}_0$ is \emph{bad} if it has at least two vertices adjacent to some
free-vertex in a 3-path in $\mathcal{P}_L$ and \emph{good} otherwise. A 3-path in $\mathcal{P}_L$ is \emph{bad} if it is adjacent to a bad 3-path in $\mathcal{P}_0$ and \emph{good} otherwise.

Extended Conditions include the following seven items:
\begin{enumerate}
\item[(E1)] For each 3-path $L_i\in \mathcal{P}\setminus \mathcal{P}^1$, at most one vertex in $L_i$ is adjacent to some vertex in $A$, i.e.,
$\mathcal{P}\setminus \mathcal{P}^1=\mathcal{P}_0\cup \mathcal{P}_1$;
\item[(E2)] No 3-path in $\mathcal{P}_M$ is adjacent to both of $A_0$-vertices and $A_1$-vertices;
\item[(E3)] No free-vertex in a 3-path in $\mathcal{P}_L$ is adjacent to a free-vertex in another 3-path in $\mathcal{P}_L$;
\item[(E4)] No free-vertex in a 3-path in $\mathcal{P}_L$ is adjacent to a free-vertex in a 3-path in $\mathcal{P}_M$;
\item[(E5)] Each 3-path in $\mathcal{P}^1$ has at most one vertex adjacent to a free-vertex in a 3-path in $\mathcal{P}_L$;
\item[(E6)] If a 3-path in $\mathcal{P}^0$ has at least two vertices adjacent to some free-vertex in a 3-path in $\mathcal{P}_L$, then all those free-vertices are from one 3-path in $\mathcal{P}_L$, i.e., each bad
    3-path in $\mathcal{P}^0$ is adjacent to free-vertices in only one bad 3-path in $\mathcal{P}_L$;
\item[(E7)] No free-vertex in a 3-path in $\mathcal{P}_L$ is adjacent to a vertex in $Z$.
\end{enumerate}

\begin{lemma}\label{hardpart}
A crucial partition of the vertex set of any given graph can be found in polynomial time.
\end{lemma}

%To avoid distraction from our main discussions, we move the proof of this lemma and the corresponding algorithm
%to the next subsection.
After obtaining a crucial partition $(A,B,Z)$, we use the following three reduction rules to reduce the graph.
In fact, Extended Conditions are mainly used for the third reduction rule and the analysis of the kernel size.

\begin{drul}\label{drule1}
If the number of 3-paths in $\mathcal{P}$ is greater than $k-|Z|/5$, halt and report it as a no-instance.
\end{drul}
Note that each $P_3VC$-set of the graph $G$ must contain at least $|Z|/5$ vertices in $Z$
by Basic Condition (B4)
and each $P_3VC$-set must contain one vertex from each 3-path in $\mathcal{P}$.
If the number of 3-paths in $\mathcal{P}$ is greater than $k-|Z|/5$, then any $P_3VC$-set of the graph has a
size greater than $k$.

\begin{drul}\label{drule2}
If $Comp_2(G[A])> |N_2(A)|$ (the condition in Corollary~\ref{lemma_corollary}) holds,
then find a good decomposition by Corollary~\ref{lemma_corollary} and reduce the instance based on the good decomposition.
\end{drul}

Reduction Rule~\ref{drule2} is easy to observe. Next, we consider the last reduction rule.
Let $B^*$ be the set of free-vertices in good 3-paths in $\mathcal{P}_L$ and
let $A^*$ be the set of $A_0$-vertices adjacent to 3-paths in $\mathcal{P}^1$.
Let $A'=A\cup B^*\setminus A^*$. By the definition of crucial decompositions, we can get that
\begin{lemma} \label{endingcondition}
The set $A'$ still induces of a graph of maximum degree 1.
\end{lemma}
\begin{proof}
Vertices in $B^*$ are free-vertices and
then any vertex in $B^*$ is not adjacent to a vertex in $A$.
Furthermore, no two free-vertices in $B^*$ from two different 3-paths in $\mathcal{P}_L$ are adjacent
by Extended Condition (E3). Since $A$ induces a graph of maximum degree 1, we know that $A\cup B^*$ induces a graph
of maximum degree 1. The set $A'=A\cup B^*\setminus A^*$ is a subset of $A\cup B^*$ and then
$A'$  induces of a graph of maximum degree 1.
\hfill \qed
\end{proof}

Based on Lemma~\ref{endingcondition}, we can apply the following reduction rule.

\begin{drul}\label{drule3}
If $Comp(G[A'])> 2|N(A')|-|N'_2(A')|$ (the condition in Lemma~\ref{lemma_lemma1} on set $A'$) holds,
then find a good decomposition by Lemma~\ref{lemma_lemma1} and reduce the instance based on the good decomposition.
\end{drul}

\medskip
Next, we assume that none of the three reduction rules can be applied and prove
that the graph has at most $5k$ vertices.

We consider a crucial partition $(A,B,Z)$ of the graph.
Let $k_1$ be the number of 3-paths in  $\mathcal{P}$.
Since Reduction Rule~\ref{drule1} cannot be applied, we know that
\begin{eqnarray}\label{rb1}
k_1\leq k-|Z|/5.
\end{eqnarray}
Since Reduction Rule~\ref{drule2} and Reduction Rule~\ref{drule3} cannot be applied, we also have the following two relations
\begin{eqnarray}\label{rb2}
Comp_2(G[A])\leq  |N_2(A)|,
\end{eqnarray}
and
\begin{eqnarray}\label{rb3}
Comp(G[A'])\leq 2|N(A')|-|N'_2(A')|.
\end{eqnarray}

By Extended Condition (E1), we know that $\mathcal{P}=\mathcal{P}_0\cup \mathcal{P}_1\cup \mathcal{P}^1=\mathcal{P}_0\cup
\mathcal{P}_L  \cup \mathcal{P}_M \cup \mathcal{P}^1$.
Let $x_1$ and $x_2$ be the numbers of good and bad 3-paths in $\mathcal{P}_L$, respectively.
Let $y_i$ ($i=0,1$) be the number of 3-paths in $\mathcal{P}_0$ with $i$ vertices adjacent to
some free-vertex in a 3-path in $\mathcal{P}_L$, and
$y_2$ be the number of 3-paths in $\mathcal{P}_0$ with at least two vertices adjacent to
some free-vertex in a 3-path in $\mathcal{P}_L$, i.e., the number of bad 3-paths in $\mathcal{P}_0$.
Let $z_1$ and $z_2$ be the numbers of 3-paths in $\mathcal{P}_M$ adjacent to only $A_0$-vertices and only $A_1$-vertices, respectively.
Let $w_1$ be the number of 3-paths in $\mathcal{P}^1$ adjacent to some free-vertex in a 3-path in $\mathcal{P}_L$
and $w_2$ be the number of 3-paths in $\mathcal{P}^1$ not adjacent to any free-vertex in a 3-path in $\mathcal{P}_L$.
We get that
\begin{eqnarray} \label{k1size}
k_1=x_1+x_2+y_0+y_1+y_2+z_1+z_2+w_1+w_2.
\end{eqnarray}
%Also by Extended Condition (E1), we know that
%\begin{eqnarray} \label{na_size}
%|N(A)|= x_1+x_2+z_1+z_2+w_1+w_2.
%\end{eqnarray}
By Extended Conditions (E1) and (E2), we know that
\begin{eqnarray} \label{na2_size}
|N(A)_2|\leq x_1+x_2+z_2.
\end{eqnarray}

Extended Condition (E6) implies the number
of bad 3-paths in $\mathcal{P}_L$ is at most the number
of bad 3-paths in $\mathcal{P}_0$, i.e.,
\begin{eqnarray} \label{x2y2}
x_2\leq y_2.
\end{eqnarray}

Each 3-path in $\mathcal{P}^1$ is adjacent to only one $A_0$-vertex. Since
$A^*$ is the set of $A_0$-vertices adjacent to 3-paths in $\mathcal{P}^1$,
we know that $|A^*|$ is not greater than $w_1+w_2$, i.e., the number of 3-paths in $\mathcal{P}^1$.
By the definition of $A'$, we know that
\begin{eqnarray} \label{nb_0}
Comp(G[A'])\geq Comp(G[A])+x_1-(w_1+w_2).
\end{eqnarray}

Next, we consider $|N(A')|$ and $|N'_2(A')|$.
Note that each 3-path has at most one vertex adjacent to vertices in $A\setminus A^*$ by
Extended Condition (E1). This property will also hold for the vertex set $A'=(A\setminus A^*)\cup B^*$.
We prove the following two relations
\begin{eqnarray} \label{nb_1}
|N(A')|\leq x_1+x_2+y_0+y_1+z_1+z_2+w_1,
\end{eqnarray}
and
\begin{eqnarray} \label{nb_2}
|N'_2(A')|\geq y_1+z_2+w_1.
\end{eqnarray}
By Extended Conditions (E1) and (E3), we know that each 3-path in $\mathcal{P}_L$ has at most one vertex in
$N((A\setminus A^*)\cup B^*)= N(A')$. By the definition of good 3-paths in $\mathcal{P}_0$, we know that
each good 3-path in $\mathcal{P}_0$ has no vertex adjacent to vertices in $A$ and has
at most one vertex adjacent to vertices in $B^*$ (which will be in a component of size 2 in $G[A']$).
There are exactly $y_1$ vertices in good 3-paths in $\mathcal{P}_0$ adjacent to vertices in $B^*$.
No vertex in a bad 3-path in $\mathcal{P}_0$ is adjacent to a vertex in  $A\cup B^*$ by the definitions of bad 3-paths and $B^*$.
Each 3-path in $\mathcal{P}_M$ has at most one vertex adjacent to $A'$ by Extended Conditions (E1) and (E4).
Only $z_2$ vertices in 3-paths in $\mathcal{P}_M$ are adjacent to vertices in $A'$, all of which
are vertices of degree-1 in $G[A']$.
No vertex in a 3-path in $\mathcal{P}^1$ is adjacent to a vertex in $A\setminus A^*\supseteq A'$ by the definition of $A^*$.
 Furthermore, each 3-path in $\mathcal{P}^1$ has at most one vertex adjacent to vertices in $B^*$ (which will be in a component of size 2 in $G[A']$) by Extended Condition (E5) and there are  exactly $w_1$ vertices in 3-paths in $\mathcal{P}^1$ adjacent to vertices in $B^*$.
%Each 3-path in $\mathcal{P}_M$ has at most one vertex adjacent to $A\setminus A^*\supseteq A'$ by Extended Condition (E4).
No vertex in $Z$ is adjacent to a vertex in $A\cup B^*$ by  Basic Condition (B3) and Extended Condition (E7).
Summing all above up, we can get (\ref{nb_1}) and (\ref{nb_2}).

Relations (\ref{rb3}), (\ref{nb_0}), (\ref{nb_1}) and (\ref{nb_2}) imply
\begin{eqnarray} \label{final1}
Comp(G[A])\leq 2(x_2+y_0+z_1+w_1)+x_1+y_1+z_2+w_2.
\end{eqnarray}
According to (\ref{rb2}) and (\ref{na2_size}), we know that
\begin{eqnarray} \label{final2}
Comp_2(G[A])\leq x_1+x_2+z_2.
\end{eqnarray}
Note that $|A|=Comp(G[A])+Comp_2(G[A])$, we get
\[\begin{array}{*{20}{l}}
{|A|}& = &{Comp(G[A]) +Comp_2(G[A])}&\\
{}& \le & 2( x_1+x_2+y_0+z_1+z_2+w_1)+ x_2+y_1+ w_2&~~~~~\mbox{by~(\ref{final1}) and (\ref{final2})}\\
{}& \le & 2( x_1+x_2+y_0+z_1+z_2+w_1)+ y_2+y_1+ w_2&~~~~~\mbox{by~(\ref{x2y2})}\\
{}& \le & 2k_1 &~~~~~\mbox{by~(\ref{k1size})}.
\end{array}\]

Note that $|B|=3k_1$ and $k_1\leq k-|Z|/5$ by (\ref{rb1}). We get that
\[\begin{array}{*{20}{l}}
{|V|}& = &|A|+|B|+|Z|\\
{}& \le & 5k_1+|Z|\le 5k.
%{}& \le & 5k.
\end{array}\]

\begin{theorem}
The parameterized 3-path vertex cover problem allows a kernel of at most $5k$ vertices.
\end{theorem}


\end{document}